\documentclass[11pt]{article}

\usepackage{graphicx}
\usepackage{amsmath,amssymb,mathtools}
\usepackage{paralist}
\usepackage{bm}
\usepackage{xspace}
\usepackage{url}
\usepackage{boxedminipage}
\usepackage{wrapfig}
\usepackage{ifthen}
\usepackage{color}
\usepackage{xcolor}
\usepackage{framed}
\usepackage[letterpaper=true,colorlinks=true,pdfpagemode=none,urlcolor=blue,linkcolor=blue,citecolor=violet,pdfstartview=FitH]{hyperref}

\usepackage{thm-restate}

\newcommand{\ignore}[1]{}

\newcommand{\Exp}{\EX}

\newcommand{\EX}{\hbox{$\mathbb E$}\xspace}

\usepackage[letterpaper]{geometry}
\usepackage{graphicx}

\usepackage[sort]{cite}
\usepackage{amssymb,amsmath, mathtools,dsfont}
\usepackage{url}
\usepackage{fullpage, prettyref}
\usepackage{boxedminipage}
\usepackage{wrapfig}
\usepackage{ifthen}
\usepackage{color,xcolor}
\usepackage{transparent}
\usepackage{relsize}
\usepackage{enumitem}
\usepackage{varwidth}
\usepackage{mdframed}
\usepackage{tcolorbox}
\usepackage[nameinlink]{cleveref}

\def\2plus{{\tt (++)}}
\def\3plus{{\tt (+++)}}
\def\4plus{{\tt (++++)}}
\def\5plus{{\tt (+++++)}}

\crefname{@theorem}{theorem}{theorems}
\crefname{Conjecture}{Conjecture}{Conjectures}
\crefname{claim}{claim}{claims}

        {\hspace*{\fill}$\Box$\par}
\newcommand{\etal}{\xspace et al.}

\newcommand{\mylabel}[2]{#2\def\@currentlabel{#2}\label{#1}}
\newlength{\algobox}
\newcommand{\E}[1]{\mathbb{E}\lb #1 \rb}
\newcommand{\Prob}[1]{\Pr\lb #1 \rb}

\newcommand{\Zp}{\mathbb Z^+}

\newcommand{\lpn}{\left(}
\newcommand{\rpn}{\right)}
\newcommand{\lbar}{\left|}
\newcommand{\rbar}{\right|}
\newcommand{\lbr}{\left\{}
\newcommand{\rbr}{\right\}}
\newcommand{\lb}{\left[}
\newcommand{\rb}{\right]}

\newcommand{\free}{\mathrm{free}}
\newcommand{\given}{\ |\ }
\newcommand{\mgiven}{\ \big|\ }
\newcommand{\bgiven}{\ \bigg|\ }
\newcommand{\thh}{^\text{th}}

\newcommand{\npColor}{{\sf Decentralized Coloring}\xspace}
\newcommand{\pColor}{{\sf Persistent Decentralized Coloring}\xspace}

\newcommand{\recolors}{\mathrm{recolors}}

\newcommand{\ind}[1]{{\mathlarger{\mathds{1}}}_{\lbr#1\rbr}}

\newcommand{\qed}{\qquad\vbox{\hrule height0.6pt\hbox{%
\vrule height1.3ex width0.6pt\hskip0.8ex
\vrule width0.6pt}\hrule height0.6pt
}\outerparskip 0pt} 

\renewcommand{\varphi}{\Phi}

\usepackage{amsthm}
\newtheorem{theorem}{Theorem}
\newtheorem{Example}{Example}
\newtheorem{Remark}{Remark}
\newtheorem{Conjecture}{Conjecture}
\newtheorem{lemma}{Lemma}

\title{On a Decentralized $(\Delta{+}1)$-Graph Coloring Algorithm}
\author{Deeparnab Chakrabarty\footnote{Dartmouth College. Email: {\tt deeparnab@dartmouth.edu}. Supported in part by CCF-1813053}\and Paul de Supinski\thanks{Work performed as part of an  undergraduate honors thesis~\cite{paul} at Dartmouth. Email: {\tt pdesupinski@yahoo.com}}}
\date{}

\begin{document}
\maketitle
\begin{abstract} 
	We consider a decentralized graph coloring model where each vertex only knows its own
	color and whether some neighbor has the same color as it. The networking community has studied this model extensively due to its
	applications to channel selection, rate adaptation, etc. 
	Here, we analyze variants of a simple algorithm of Bhartia et al. [Proc., ACM MOBIHOC, 2016]. In particular, we introduce a variant which requires only $O(n\log\Delta)$ expected recolorings that generalizes the coupon collector problem. Finally, we show that the $O(n\Delta)$ bound Bhartia et al. achieve for their algorithm still holds and is tight in adversarial scenarios.
\end{abstract}
  
\section{Introduction}
It is well known that an undirected graph $G=(V,E)$ with maximum degree $\Delta$ can be properly vertex-colored using $(\Delta{+}1)$ colors.
The simple ``greedy'' algorithm makes one pass over the nodes, giving each node one of the colors not currently used by its neighbors.

Motivated by applications to channel selection for access points, Bhartia\etal~\cite{iq} investigate highly constrained decentralized algorithms for graph coloring. In their setting, the only information a vertex knows at any time is its own color and whether at least one adjacent vertex has the same color. 
In the case of networking applications, nodes correspond to access points, colors correspond to transmission channels, and edges correspond to whether two access points interfere with each other when transmitting in the same channel. Fittingly, an access point only knows its own channel and whether some neighbor is using the same channel, which can be inferred from the resulting packet loss. Accordingly, the networking community has studied this model extensively ~\cite{duffy, duffy2, galan, fitz, leith, roughgarden, checco, iq}. This model is sometimes called the {\em conflict detection model}~\cite{roughgarden}.

In the stylized setting below, we describe the algorithm proposed by Bhartia\etal~\cite{iq}, arguably the simplest and most natural one for the model. 
This algorithm proceeds over time and maintains a coloring $\chi_t: V \to \{1,2,\ldots, \Delta+1\}$. A vertex $v$ is \textit{conflicted} at time $t$ if 
there is some neighbor $u$ of $v$ such that $\chi_t(u) = \chi_t(v)$.

\begin{mdframed}[backgroundcolor=gray!20,linecolor=white]
	\noindent
	\underline{\npColor} ($G = (V,E)$)
	\begin{enumerate}
		\item Initially, every vertex $v$ chooses a color $\chi_0(v)$ at random from $\{1,2,\ldots,\Delta+1\}$. \label{alg:1}
		\item At each time $t$, a vertex $v$ is chosen uniformly at random among all conflicted vertices. \label{alg:2}
		\item $v$ changes its color to a random color in $\{1,2,\ldots, \Delta+1\}$. \label{alg:3}
		\item Steps~\ref{alg:2} and~\ref{alg:3} repeat until there are no conflicted vertices.
	\end{enumerate}
\end{mdframed}
In the decentralized model, Bhartia\etal\cite{iq}~implement this algorithm by having vertices wait random amounts of time between recolors (Steps ~\ref{alg:3} and ~\ref{alg:2}). They also prove that the algorithm converges to a proper $(\Delta{+}1)$-coloring in $O(n\Delta)$ expected recolorings. However, our results, which now summarize, strongly suggest that this bound is \textit{not} tight. As an introduction, consider the special case when the graph is a clique, which turns out to be trivial.

\begin{Example}\label{ex:clique}
	Let $H_k \coloneqq \sum_{i=1}^k \frac 1 i $ be the $k\thh$ harmonic number. On $K_n$, the clique of $n$ vertices, \npColor converges to a $(\Delta{+}1)$-coloring in exactly $nH_n = \Theta(n\log n)$ expected recolorings.
		
	To see this, observe that \npColor is essentially the coupon collector process on cliques.
	That is, all $n$ vertices require different colors, and no color once in the graph can ever be fully removed from the graph. Hence, the process terminates when all $n$ colors have been chosen exactly once. Thus, the number of recolorings (including the initial $n$ from Step~\ref{alg:1}) is precisely the number of draws to obtain all $n$ coupons, whose expected value is well known to be $nH_n$. 
\end{Example}

\noindent
Our first contribution is to introduce a variant of \npColor which is easier to analyze. The sole difference is the while-loop of Step~\ref{alg2:3}
\begin{mdframed}[backgroundcolor=gray!20,linecolor=white]
	\noindent
	\underline{\pColor} ($G = (V,E)$)
	\begin{enumerate}
		\item Initially, every vertex $v$ chooses a color $\chi_0(v)$ at random from $\{1,2,\ldots,\Delta+1 \}$. \label{alg2:1}
		\item At each time $t$, a vertex $v$ is chosen uniformly at random among all conflicted vertices. \label{alg2:2}
		\item While $v$ is conflicted, it keeps changing to a random color in $\{1,2,\ldots, \Delta+1\}$. \label{alg2:3}
		\item Steps~\ref{alg:2} and~\ref{alg:3} repeat until there are no conflicted vertices.
	\end{enumerate}
\end{mdframed}

As our second contribution, we prove the following theorem in~\Cref{sec:proof}.
\begin{restatable}{theorem}{main}\label{thm:main}
	The \pColor algorithm converges to a proper $(\Delta{+}1)$-coloring in $O(n\log \Delta)$ expected recolorings.
\end{restatable}

For our final two contributions, we analyze adversarial variants of the two algorithms. For either algorithm, if we allow an adversary to choose the initial coloring $\chi_0$ in Step~\ref{alg:1}, we say the algorithm uses an \textit{adversarial start}. Similarly, if we allow an adversary to choose the conflicted vertex in Step~\ref{alg:2}, we say the algorithm uses an \textit{adversarial order}. 
\begin{Remark}
	In adversarial order \npColor, the adversary could choose vertices so as to mimic \pColor.
	Therefore, a lower bound for \pColor implies a lower bound for adversarial order \npColor.
\end{Remark}

It may be interesting to ponder whether the algorithms even converge given one or both modifications.

As our third contribution, in~\Cref{sec:drift} we show that in fact \npColor still only requires $O(n\Delta)$ expected recolorings in the adversarial start, adversarial order case. 
\begin{restatable}{theorem}{ub}\label{thm:ub}
	Even with an adversarial initial coloring $\chi_0$ in Step~\ref{alg:1} and an adversarial choice of conflicted vertices in Step~\ref{alg:2},
	the \npColor algorithm converges to a proper $(\Delta{+}1)$-coloring in $O(n\Delta)$ expected recolorings.
\end{restatable}
In other words, we achieve the same bound as Bhartia\etal\cite{iq}, whose proof would only yield an $O(n\Delta^2)$ bound in this case, while forgoing the randomness from all but Step~\ref{alg:3}. 

Encouraged by~\Cref{thm:main} and the clique example, one may simply conjecture that all variants require only $O(n\log\Delta)$ recolorings. However, our fourth contribution is a counterexample we give in~\Cref{sec:counter} showing that~\Cref{thm:ub} is tight in certain cases, even when the order of vertices is still random.

\begin{theorem}\label{thm:counter}
	With an adversarial initial coloring $\chi_0$ in Step~\ref{alg:1} and random choice of vertices in Step~\ref{alg:2}, \pColor requires $\Omega(n\Delta)$ expected recolorings in the worst case.
\end{theorem}

Notably, this counterexample will not apply to random order \npColor. So, finally, we offer the following conjecture, which motivated this research, but whose proof eludes us.

\begin{Conjecture}\label{conj:NP}
	The \npColor algorithm finds a proper $(\Delta{+}1)$-coloring in $O(n\log \Delta)$ recolorings.
\end{Conjecture}

\section{Persistent Decentralized Coloring}
\subsection{Adversarial Start, Random order}\label{sec:counter}~\\
As a warmup, we begin with the counterexample proving~\Cref{thm:counter}.

\begin{proof}
	\begin{figure}[ht!]
		\centering
		\includegraphics[scale = .85]{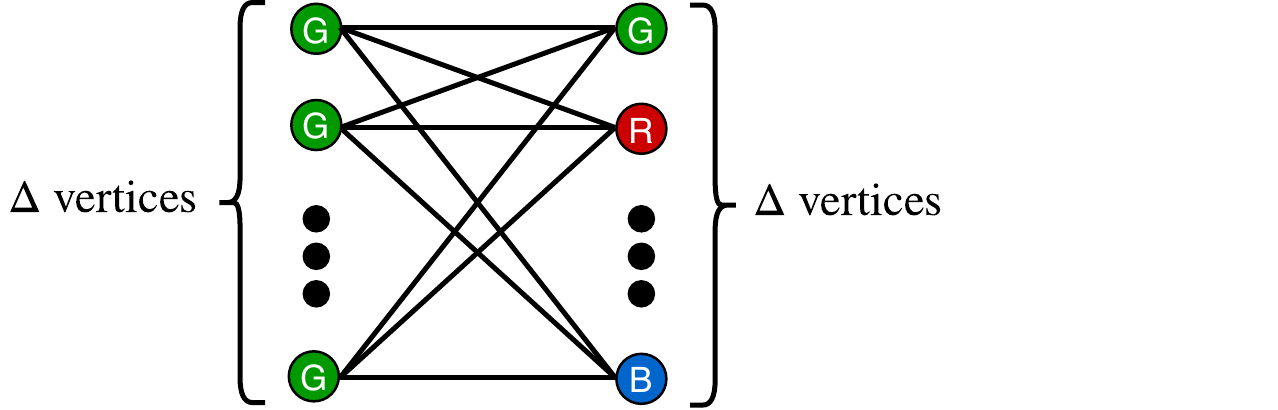}
		\caption{A bad initial coloring for \pColor}
		\label{fig:badex}
	\end{figure}
							
	Consider the complete bipartite graph $K_{\Delta,\Delta}$. 
	Suppose we initially color every left side vertex with the same color, green, and use $\Delta$ colors on the right half, including the color green (see~\Cref{fig:badex}). In this configuration, every left side vertex is conflicted, and there is only one conflicted right side vertex $v$. Furthermore, the left side vertices have only one free color, and hence would each expect to recolor $\Delta{+}1$ times if selected. 
						  
	On average, we recolor half of the left side vertices before fixing the right side vertex $v$ (at which point the process terminates). Hence, the total expected run time is $\Omega(n\Delta)$. 
\end{proof}

\subsection{Random Start, Random Order}\label{sec:proof}~\\
In this section, we prove~\Cref{thm:main}, restated here for convenience.
\main*
\noindent Our strategy is to localize our analysis to an arbitrary vertex $v$ and then proceed by coupling. In particular, we couple $v$ with an arbitrary vertex from the clique of size $\deg(v)$, whose behavior we understand from the coupon collector coupling.

Now, fix $v$ to be an arbitrary vertex.
\begin{lemma}\label{lem:main}
	The expected number of recolorings of $v$ is $\leq H_{\deg(v)}$.
	The expectation is over (a) the random initial coloring $\chi_0$ in Step~\ref{alg2:1}, (b) the random order in which conflicted vertices are picked in Step~\ref{alg2:2}, and (c) the randomness in the recoloring in Step~\ref{alg2:3}.
\end{lemma}
\Cref{lem:main} immediately implies \Cref{thm:main} because the expected number of recolorings of any vertex is thus $\leq H_{\Delta} = O(\log \Delta)$.
%
Again, it is essential to take the expectation over the random initial colorings $\chi_0$, otherwise \Cref{fig:badex} would act as a counterexample.

\begin{proof}
	We begin by setting some notation. As usual, we let $\Gamma(v)$ denote the neighborhood of $v$, not including $v$ itself. For simplicity let $d := \deg(v)$, and $D := \Delta+1$. Next, observe that Step~\ref{alg2:2} in \pColor can be simulated by first selecting a random permutation $\pi$ of the vertices, and then selecting the conflicted vertices in the $\pi$ order. This is valid because no vertex can ever be chosen twice in Step~\ref{alg2:2} (unlike in \npColor). Recall that $\chi_0$ is the initial coloring of the graph. 
																														
	Given $\pi$ and $\chi_0$, we define $\recolors_{\pi,\chi_0}(v)$ to be the random variable indicating the number of times $v$ recolors given that the initial coloring was $\chi_0$ and the order of vertices was $\pi$. Note that the randomness of $\recolors_{\pi,\chi_0}(v)$ arises solely from Step~\ref{alg2:3} of \pColor. Our goal is to bound $\Exp_{\pi, \chi_0} \Exp [\recolors_{\pi,\chi_0}(v)]$, the expected number of recolors of $v$, averaged over all $\pi$ and $\chi_0$.
																														
	Given $\pi$, let $B_\pi(v)$ and $A_\pi(v)$ denote the subsets of $\Gamma(v)$ which come {\em before} and {\em after} $v$ in the permutation $\pi$, respectively. Let $\free_{\pi, \chi_0}(v)$ be the random variable denoting the number of colors {\em not} used by $\Gamma(v)$ when $v$ begins recoloring, given that the initial coloring is $\chi_0$ and the order is $\pi$. 
																																												
	We now proceed with a coupling argument. Let $v_1, v_2, \ldots, v_d$ be an arbitrary labeling of $\Gamma(v)$. Let $K$ denote the $(d{+}1)$-sized clique with vertices arbitrarily labeled $w,w_1,\ldots, w_d$. We now consider the \pColor process on $K$, but using $D$ colors even if $d < \Delta$. We couple the initial coloring $\chi'_0$ with $\chi_0$, and the order $\pi'$ with $\pi$. In particular, let $\chi'_0(w) = \chi_0(v)$ and $\chi'_0(w_i) = \chi_0(v_i)$ for all $1\leq i\leq d$. Let the order $\pi'$ of $\{w,w_1,\ldots, w_d\}$ be equal to the order $\pi$ restricted to $\{v,v_1,\ldots, v_d\}$, with the same vertex pairings as before. To be clear, $A_{\pi'}(w) = \{w_i: v_i \in A_\pi(v)\}$, and $B_{\pi'}(w)=\{ w_i : v_i \in B_\pi(v)\}$. Finally, $\free_{\pi', \chi'_0}(w)$ is the number of colors \textit{not} used by $\Gamma(w)$ when $w$ begins recoloring.
																													
	\begin{lemma}\label{clm:ub}
		For any $\chi_0$ and $\pi$, we have $\Exp[\recolors_{\pi,\chi_0}(v)] \leq \Exp[\recolors_{\pi',\chi'_0}(w)]$.
	\end{lemma}
	\begin{proof}
		Observe that $v$ has to recolor iff\footnote{We use the notation $\chi_0(S) := \{ \chi_0(z) : z \in S\}$.} $\chi_0(v) \in \chi_0\lpn A_\pi(v)\rpn$. This is because each vertex of $B_\pi(v)$ necessarily fixes to a different color than $v$'s color (which is still $\chi_0(v)$) prior to $v$'s turn. Thus, the only way $v$ could be still be conflicted is if $\chi_0(v) \in  \chi_0\lpn A_\pi(v)\rpn$.
		Similarly, $w$ has to recolor iff $\chi'_0(w) \in \chi'_0\lpn A_{\pi'}(w)\rpn$. 
		But $\chi_0\lpn A_\pi(v) \rpn = \chi'_0\lpn A_{\pi'}(w)\rpn$ and $\chi_0(v) = \chi'_0(w)$, so $\recolors(v) > 0$ iff $\recolors(w) > 0$ under our coupling.
																																																										
		Next, observe that $\free_{\pi, \chi_0}(v) \geq \free_{\pi', \chi'_0}(w)$ under our coupling.	
		To see this, note that $\free_{\pi, \chi_0}(v)$ is equivalently the total number of colors, $D$, minus the number of different colors used by $\Gamma(v)$ when $v$ begins recoloring. In the case of the $(d{+}1)$-clique $K$, when $w$ begins recoloring, we can guarantee that the set of colors used by $B_{\pi'}(w)$ has size $|B_{\pi'}(w)|$ and is disjoint from the set of colors used by $A_{\pi'}(w)$, because the vertices are all connected. So $A_{\pi}(v)$ and $A_{\pi'}(w)$ use the same colors, and $B_{\pi'}(w)$ uses at least as many additional colors as $B_{\pi}(v)$.
																																																										
		Finally, note that $\recolors_{\pi,\chi_0}(v)$ is either $0$ or the geometric random variable whose probability parameter is $\free_{\pi,\chi_0}(v)/D$,
		with a similar statement for $w$. Thus,
		\begin{align*}
			  \E{\recolors_{\pi,\chi_0}(v)} & \quad=\E{\ind{\recolors_{\pi,\chi_0}(v) > 0} \cdot \frac{D}{\free_{\pi,\chi_0}(v)}}         \\
			                                & \quad\leq \E{\ind{\recolors_{\pi',\chi'_0}(w) > 0} \cdot \frac{D}{\free_{\pi',\chi'_0}(w)}} 
									   	  \quad= \E{\recolors_{\pi',\chi'_0}(w)}.                                                     
		\end{align*}
\end{proof}
																													
	\begin{lemma}\label{clm:clq}
		For any $1 \leq i \leq d$, we have \[\Exp_{\pi', \chi'_0}\Exp[\recolors_{\pi',\chi'_0}(w)] = \Exp_{\pi', \chi'_0}\Exp[\recolors_{\pi',\chi'_0}(w_i)]\]
	\end{lemma}
	\begin{proof}
		This is by symmetry of the clique. 
		%
		It may be instructive to point out that the randomness of $\chi'_0$ and $\pi'$ are both necessary. For example, if we fix an initial coloring $\chi'_0$, then no initially happy vertex ever recolors. If we fix an ordering $\pi'$, then the last vertex never recolors.
	\end{proof}
	\noindent
	\Cref{clm:ub} and \Cref{clm:clq} together imply that 
	\begin{equation}
		  \Exp_{\pi, \chi_0} \Exp [\recolors_{\pi,\chi_0}(v)]
		  \leq ~\Exp_{\pi', \chi'_0} \Exp[\recolors_{\pi',\chi'_0}(w)]
		  = ~\frac{1}{d+1} \sum_{x\in K} \Exp_{\pi', \chi'_0} \Exp[\recolors_{\pi',\chi'_0}(x)].\label{eq:couple} 
	\end{equation}
	We already know how to bound the sum in~\Cref{eq:couple}. It is precisely the expected total number of recolors of \pColor on a $(d{+}1)$-clique, but with $D$ colors available. In fact, for {\em any} $\pi'$ and $\chi'_0$, we know that 
	\begin{align}
		  \sum_{x\in K} \Exp[\recolors_{\pi',\chi'_0}(x)]
		  & \quad\leq \frac{D}{D-1} + \cdots + \frac{D}{D - d}\label{eq:middle} \\
		  & \quad\leq (d+1)H_{d}.\label{eq:cc}                                  
	\end{align}
	This follows from another coupon collector argument. Here, \Cref{eq:middle} represents the time to collect $d{+}1$ out of $D$ coupons, and \Cref{eq:cc} represents the time to collect $d{+}1$ out of $d{+}1$ coupons. (We can omit the leading 1 in both sums because the initial coloring includes at least a first color.) Simple manipulation shows that $\frac{D}{D -i} \leq \frac{(d+1)}{d+1-i}$ when $i\geq0$, because $D \geq d+1$. Thus, the second inequality holds. 
	Together, \Cref{eq:couple} and \Cref{eq:cc} imply \Cref{lem:main}.
	%
	%
	%
	%
\end{proof}
\noindent
Since $v$ was an arbitrary vertex, \Cref{thm:main} follows from \Cref{lem:main} and linearity of expectation.

\section{Decentralized Coloring with Adversarial Start and Order}\label{sec:drift}
In this section, we prove~\Cref{thm:ub}, restated here for convenience. 
\ub*

Recall that $D := \Delta{+}1$.
Before we begin, observe that this bound is fairly trivial for \pColor:
in Step~\ref{alg2:3}, there is always at least a $\frac{1}{D}$ chance that recoloring satisfies the chosen vertex, implying that each vertex recolors at most $D$ times in expectation. However, in \npColor, there is no similar concept of vertices becoming fixed. Instead, our strategy is to analyze the rate at which \npColor drifts toward convergence. 

One way to analyze drift is with a potential argument. This entails defining a potential function which monotonically changes in expectation with each iteration of the algorithm. 

In our case, we define a potential function $\Phi$ on graph colorings $\chi$ such that $\chi$ is valid iff $\Phi(\chi)$ is some value $\lambda$. Then, we show that $\Exp[\Phi(\chi_t)]$ converges toward $\lambda$ monotonically in expectation at a bounded rate as $t$ increases. Indeed, this is the approach of Bhartia\etal~\cite{iq}, who choose $\Phi(\chi)$ to denote the number of conflicted edges in $G$ with respect to $\chi$, in which case $\lambda = 0$. (An edge is conflicted iff its end points have the same color.)
It is easy to show that the expected number of conflicted edges decreases by at least $1/D$ with each recoloring. If the initial coloring $\chi_0$ is random, then it is easy to see that $\Exp[\Phi(\chi_0)] = O(n)$, which in turn implies\footnote{To make this formal, one needs to use a stopping theorem which Bhartia\etal~\cite{iq} do not explicitly mention. We prove and use such a theorem explicitly.} that the expected number of recolorings is $O(n\Delta)$.
Unfortunately, with an adversarial start, this particular argument only implies an $O\lpn n \Delta^2\rpn$ bound because there can be  $\Omega\lpn n\Delta\rpn$ conflicted edges initially.

The other obvious choice for $\varphi\lpn \chi \rpn$ is the number of conflicted vertices in $G$ under $\chi$. 
However, we can concoct examples where we would actually expect $\varphi\lpn \chi_t\rpn$ to increase, given an adversarial selection. For example, if we recolor $v$ in \Cref{posDrift}, the number of conflicted vertices increases (additively) by 1/4, on average.

\begin{figure}[ht!]
	\centering
	\includegraphics[scale = .3]{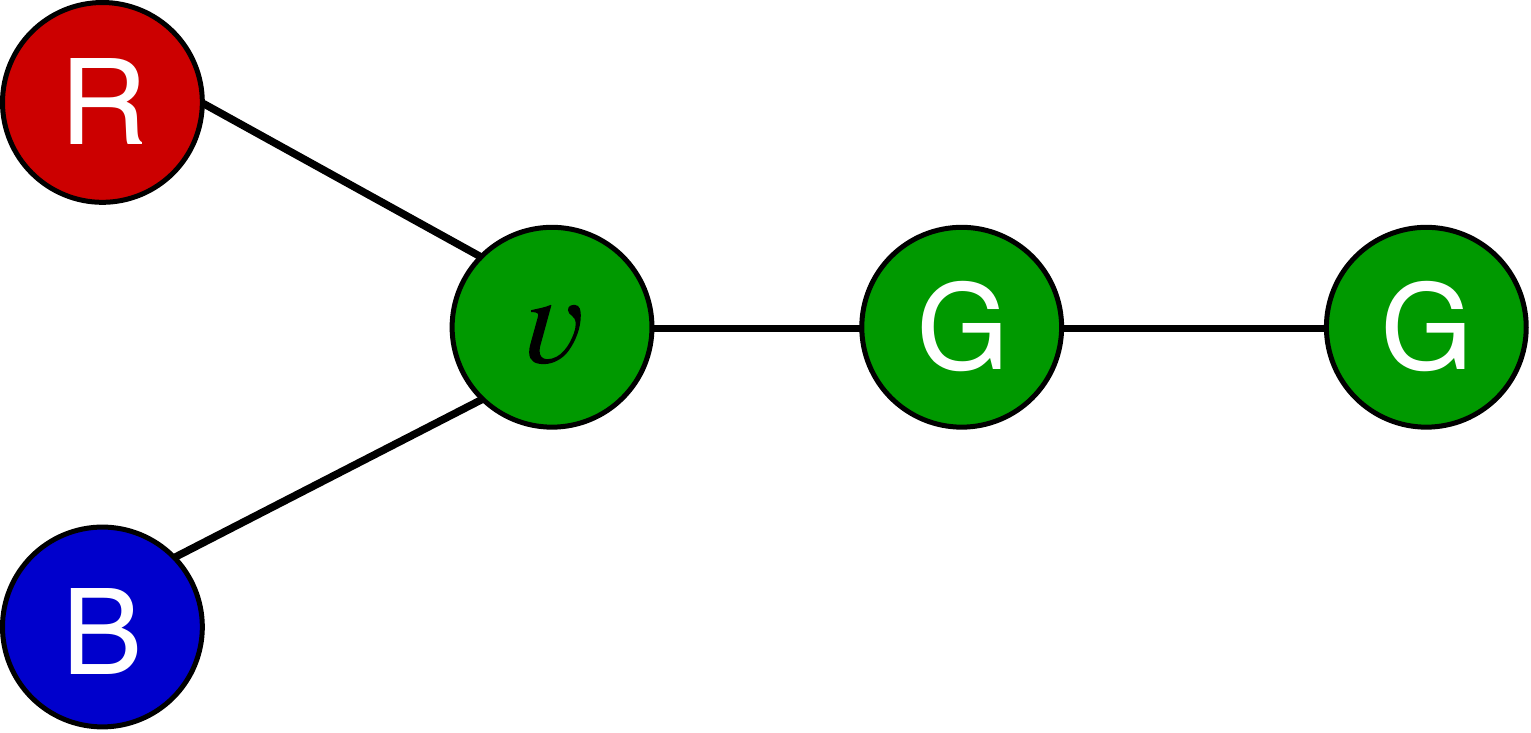}
	\caption{An example of a graph in which the number of conflicted vertices would be expected to increase, given an adversarial selection. There are 4 colors available: $\mathtt{R,G,B}$ and $\mathtt Y$. If $v$, whose color is $\mathtt{G}$, recolors to $\mathtt{G}$, then the number of conflicted vertices remains the same. If $v$ recolors to $\mathtt{Y}$, then the number decreases by $1$. However, if $v$ recolors to $\mathtt{R}$ or $\mathtt{B}$, then the number \textit{increases} by $1$.} \label{posDrift}
\end{figure}

To achieve our $O(n\Delta)$ upper bound, we define $\Phi$ such that $\Phi(\chi)$ is the number of \textit{monochromatic connected components} in $G$ under $\chi$. That is, $\Phi(\chi)$ is the number of connected components induced by the vertices of the same color, taken over all colors.
For example, in \Cref{posDrift}, there are three monochromatic components. 
Note that $\chi$ is a proper coloring iff $\Phi(\chi) = n$, because the monochromatic components all need to be singletons.

\begin{lemma}\label{lem:potn}
	For $t > 0$, let $\chi_t$ be the coloring after the $t^{th}$ recoloring in \npColor. Then,
	\[
		\Exp[\Phi(\chi_t) - \Phi(\chi_{t-1})~|~\chi_{t-1} ~\textrm{\em invalid}~] \geq \frac{1}{D},
	\]
	where the expectation is over the randomness in  Step~\ref{alg:3} for the vertex chosen at the $t^{th}$ recoloring.
\end{lemma}
\begin{proof}
																														
	Let $v$ be the vertex which recolors at time $t$ (so that $\chi_t$ has the new color of $v$), where $v$ could be chosen adversarially.
	Given a color $c \in \{1,2,\ldots, D\}$ define $m_t(c)$ to be the number of monochromatic components with respect to $\chi_t$ of color $c$ which contain \emph{at least} one vertex from $\Gamma(v)\cup \{v\}$. We exclude components of color $c$ which have no vertices in $\Gamma(v)\cup \lbr v \rbr$ because $v$'s recoloring cannot affect those components. Similarly define $m_{t-1}(c)$. 
	Observe that
	\begin{align}
		\Phi(\chi_t) - \Phi(\chi_{t-1}) = \sum_{c \in [D]} \left(m_t(c) - m_{t-1}(c)\right).\label{eq:subtract} 
	\end{align}
	With this, the proof follows from three observations. First, $m_t(\chi_t(v)) = 1$ because all adjacent components of $v$'s color are connected through $v$.
	Second, for any {\em other} color $c\neq \chi_t(v)$, we have $m_{t}(c) - m_{t-1}(c) ~\geq~ 0$ (meaning improvement). This is because recoloring $v$ to a color besides $c$ cannot possibly create any new paths of color $c$. Third, $\sum_{c \in [D]} m_{t-1}(c) \leq \Delta$. This is because $v$ is conflicted with respect to $\chi_{t-1}$ and hence has the same color as one of its neighbors. 
																														
																														\noindent
	Therefore,
	\[
	 \E{\Phi(\chi_t) - \Phi(\chi_{t-1})\given\chi_{t-1}~\textrm{invalid}} 
	= ~\E{\sum_{c \in [D]} \left(m_t(c) - m_{t-1}(c)\right)\bgiven\chi_{t-1}}
	\]
	\noindent
	which in turn evaluates to 
	\begin{align}
		&\E{\sum_{c \in [D]} \left(m_t(c) - m_{t-1}(c)\right)\bgiven\chi_{t-1}}\notag \\
		  & ~~~~~~~~~~~~~ = \sum_{c \in [D]} \Big(\Pr[\chi_t(v) = c]\cdot \left(m_t(c) - m_{t-1}(c)\right)
+ \Pr[\chi_t(v) \neq c]\cdot (m_t(c) - m_{t-1}(c))  \Big)                \label{eq:a}   \\
		  & ~~~~~~~~~~~~~\geq \sum_{c \in [D]} \Pr[\chi_t(v) = c]\cdot \left(1 - m_{t-1}(c)\right)   \ =\frac{1}{D} \cdot \sum_{c \in [D]} \left(1 - m_{t-1}(c)\right)\label{eq:b}                     \\
		  & ~~~~~~~~~~~~~\geq  1 - \frac{\Delta}{D} ~~~=~~~ \frac{1}{D}.\label{eq:c}                                      
	\end{align}
	Recall that the expectation is over the random recoloring of $v$ at time $t$. 
	We lose the conditioning on $\chi_{t-1}$ in ~\Cref{eq:a} because the new color is independently random. Also note that $m_{t-1}$ is fixed once we know $\chi_{t-1}$. 
	\Cref{eq:a} follows from \Cref{eq:subtract}. \Cref{eq:b} follows from the first two observations mentioned after \Cref{eq:subtract}.
	\Cref{eq:c} follows from the third observation and the fact that $D = \Delta+1$.
	%
	%
	%
	%
\end{proof}

%

\noindent
Using~\Cref{lem:potn} to prove that the expected stopping time $\tau(G)$ of our process is $O(n\Delta)$ requires one more theorem.
In particular, we use the following adjusted version\footnote{We don't claim novelty here; a theorem similar to \Cref{stopping} may exist elsewhere.} of Wald's equation~\cite{wald}. We also note that this allows one to provide a formal proof of the Bhartia\etal~\cite{iq} claim.
\begin{lemma}\label{stopping}
	Let $\varphi$ be a real-valued function of colorings for a graph $G$ such that $\chi$ is valid iff $\varphi\lpn \chi \rpn = \lambda$, for some constant $\lambda$. Let $\chi_t$ be the state of $\chi$ after $t$ recolorings. 
	Let $\tau(G)$ be the random number of recolorings required to produce a valid coloring of $G$.
	
	If $\E{ \lbar \lambda - \varphi\lpn \chi_{t-1} \rpn\rbar - \lbar \lambda - \varphi \lpn \chi_{t}\rpn\rbar \mgiven \chi_{t-1} \text{ \emph {invalid}}}~\geq~C$ for some positive constant $C$, then $\E{\tau\lpn G \rpn } \leq \E{\lbar \lambda - \varphi\lpn \chi_0\rpn\rbar }/C$.
\end{lemma}

We first show how \Cref{lem:potn} and \Cref{stopping} imply \Cref{thm:ub}. Recall that $\chi$ is valid iff $\varphi(\chi) = n$. We have
\begin{equation*}
	  \E{\lbar n - \varphi\lpn \chi_{t-1}\rpn\rbar - \lbar n - \varphi\lpn \chi_t\rpn\rbar \bgiven \chi_{t-1} \text{ invalid}} 
	 = \E{\varphi\lpn \chi_t\rpn - \varphi\lpn \chi_{t-1}\rpn \mgiven \chi_{t-1} \text{ invalid}}                  \geq \frac{1}{D}.                                                                                                   
\end{equation*}
Because $1\leq \varphi(\chi) \leq n$, we have $\E{ \lbar n - \varphi\lpn \chi_0 \rpn\rbar} \leq n - 1$. Hence, $\E{\tau(G)} \leq (n-1)D$. This completes the proof of~\Cref{thm:ub}. \qed
\noindent
We now prove \Cref{stopping}.
\begin{proof} 	
	For each $t \in \Zp$, let $$Z_t \coloneqq \begin{cases}
	\lbar \lambda - \varphi\lpn \chi_{t-1} \rpn\rbar - \lbar \lambda - \varphi \lpn \chi_{t}\rpn\rbar & t \leq \tau(G) \\
	0 & \text{otherwise.}\end{cases}$$ By assumption, $\varphi\lpn\chi_{\tau(G)}\rpn = \lambda$. Hence, $\sum_{t = 1}^{\tau(G)} Z_t$ telescopes to $\lbar \lambda -  \varphi \lpn \chi_{0}\rpn\rbar$. From here, we see that
	\begin{equation}
		 \E{\lbar\lambda - \varphi\lpn \chi_0\rpn\rbar} = \E{\sum_{t=1}^{\tau(G)} Z_t} = \E{ \sum_{t=1}^\infty Z_t \cdot \ind{\tau(G) \geq t}}.\label{initialBound} 
	\end{equation} 
	Next, we prepare to apply the following analogue of linearity of expectation for infinite sums.
		
	\begin{theorem}\label{infLinearity}
		~\\\textsc{(Infinite Linearity of Expectation \cite{probAndComputing})}\\\\Let $X_1, X_2, \ldots$ be random variables. If $\sum_{t=1}^\infty \E{\lbar X_t\rbar}$ converges, then
		$$\E{ \sum_{t=1}^\infty X_t} = \sum_{t=1}^\infty \E{X_t}.$$
	\end{theorem}
	%
	%
		
		
	To apply \Cref{infLinearity}, we need to show that $\sum_{t=1}^\infty \E{\lbar Z_t \cdot \ind{\tau(G) \geq t}\rbar}$ converges. Observe that $\varphi$ must be bounded, because it is real-valued and there are only finitely many possible colorings of $G$. Thus,  $\lbar Z_t \rbar \leq \rho$ for some constant $\rho$. Hence,
	\begin{equation}
		\sum_{t=1}^\infty \E{\lbar Z_t \cdot \ind{\tau(G) \geq t}\rbar} 
= \sum_{t=1}^\infty \E{\lbar Z_t\rbar \cdot \ind{\tau(G) \geq t}}
\leq \rho\sum_{t=1}^\infty\E{\ind{\tau(G) \geq t}}
= \rho \cdot \E{\tau(G)}. \label{eq:terms}                    
	\end{equation}
  Trivially, we can bound $\E{\tau(G)}$ by $nD^n$, because at worst we need to select the lone satisfying color for $n$ consecutive vertices, which can be cast as a geometric random variable with probability $\lpn\frac{1}{D}\rpn^n$ that uses at most $n$ recolors per trial. Hence, $\E{\tau(G)}$ is finite. 

  \noindent
  Because \Cref{eq:terms} has positive terms and is bounded above, it indeed converges. Thus, we have
	\begin{align*}
		\E{ \sum_{t=1}^\infty Z_t \cdot \ind{\tau(G) \geq t}}  & \quad= \sum_{t=1}^\infty \E{Z_t \cdot \ind{\tau(G) \geq t}} &   & \hspace{-1cm} (\Cref{infLinearity}) \\
		&\quad= \sum_{t=1}^\infty \E{Z_t \given \tau(G) \geq t} \cdot \Prob{\tau(G) \geq t} \\
		&\quad\geq C\sum_{t=1}^\infty \Prob{\tau(G) \geq t} ~=~ C \cdot \E{\tau(G)}.
	\end{align*}
	With \Cref{initialBound}, we have $\E{\tau(G)}\leq\E{\lbar\lambda - \varphi\lpn \chi_0\rpn\rbar} / C$, which completes the proof of \Cref{stopping}.
\end{proof}
\section{Related Work}

The decentralized model of graph coloring has been studied extensively~\cite{duffy, duffy2, galan, fitz, leith, roughgarden, checco, iq}. Most of this work appears in the networking literature, motivated by the need to minimize the communication between nodes.
We mention two works which are most similar to our work, and comment on the differences. The first work is by Motskin\etal~\cite{roughgarden} which considers an algorithm very similar to \npColor, except that all conflicted vertices simultaneously randomly recolor. 
It is not too hard to show that a $\frac{1}{\Delta{+}1}$-fraction of the nodes become happy in each round, and therefore, $O(\Delta \log n)$-rounds suffice with high probability. Still, this leads to $O(n\Delta)$ recolorings, which is no better than what Bhartia\etal~\cite{iq} achieve (note that the first random recoloring constitutes a random start). The second work is by Checco and Leith~\cite{checco}, which itself generalizes works by Barcelo\etal~\cite{barcelo} and Duffy\etal~\cite{duffy,duffy2}, where again all conflicted vertices recolor simultaneously, but according to a distribution that evolves with time. Their algorithms, which also converge in $O(n\log n)$ rounds, are robust to changes in the graph.

In this paragraph, we describe other decentralized models which are unrelated to the model we study but may be interesting to the reader.
Synchronous graph coloring in minimal number of rounds arises in {\em distributed computing}. Unlike our decentralized setting, this model, first defined in Linial's seminal paper~\cite{linial}, allows nodes to pass messages among each other, and the number of rounds
is one key complexity parameter. Johansson~\cite{johansson} obtains an $(\Delta{+}1)$-coloring in $O(\log n)$ rounds using a simple algorithm is similar to, and indeed inspired by, Luby's MIS algorithm~\cite{luby}.
This was recently improved by Harris, Schneider and Su~\cite{HSS16} to a $O(\sqrt{\log n})$-round algorithm, and more recently to a $O(\mathrm{polyloglog}~ n)$-round algorithm by Chang \etal\cite{pettie}.
$(\Delta{+}1)$-colorings have also recently been considered in the streaming model by Assadi \etal \cite{assadi} where edges stream in and there is only $\tilde{O}(n)$-space available to help maintain a $(\Delta{+}1)$-coloring. The same paper also gives algorithms in the graph-query and MPC  (massively parallel computation) models. $(\Delta{+}1)$-coloring has also recently been considered in the {\em dynamic graph model} where edges may be added or deleted and the objective is to maintain a $(\Delta{+}1)$-coloring with quick updates.
Bhattacharya\etal~\cite{deepcDistrib} describe a randomized algorithm with $O(\log n)$ amortized update time which has very recently been improved upon by Henzinger and Peng~\cite{HP19}, and Bhattacharya\etal~\cite{Bhatt19}.

%
%

\section{Conclusion and Discussion}

In this paper we considered variants of two decentralized graph coloring algorithms: \npColor, introduced by Bhartia\etal~\cite{iq}, and \pColor, our proposed modification. Beyond the \pColor algorithm itself, we produced three primary contributions:
\begin{itemize}
	\item Adversarial start, random order \pColor requires $\Omega(n\Delta)$ expected recolorings in the worst case. 
	\item Random start, random order \pColor requires only $O(n\log\Delta)$ expected recolorings.
	\item Adversarial start, adversarial order \npColor requires only $O(n\Delta)$ recolorings.
\end{itemize}

We proved the first result with a counterexample involving a bipartite graph, the second through coupling and generalization to the coupon collector problem, and the third through analysis of an interesting potential function. We also note that our stopping theorem may be extensible to other stochastic processes that tend to drift toward convergence. Lastly, we formalized the conjecture that the $O(n\log\Delta)$ bound holds for random start, random order \npColor.

It is perhaps instructive to remark that \npColor resembles, in spirit, the celebrated Moser-Tardos~\cite{mt} randomized algorithm for finding satisfying assignments to CSPs obeying the Lov\'asz Local Lemma. Indeed, if we tweak the \npColor algorithm so that we pick a conflicted {\em edge} in every round and randomly recolor its endpoints, we get the Moser-Tardos algorithm. However, $(\Delta{+}1)$-coloring is outside the LLL regime: the probability of a bad event (of getting an unhappy edge) is $p = \frac{1}{\Delta{+}1}$, and the degree $d$ of the dependency graph is $2\Delta{-}2$ (a single edge is independent of all but its $2\Delta{-}2$ neighboring edges), giving $pd \approx 2$, which causes the general analysis of~\cite{mt} to break down.

Our hope is that our new understanding of \pColor may lead to a proof of \Cref{conj:NP}. 
A first question to answer may be whether \pColor always requires more expected recolorings than \npColor on any particular graph. We also wonder whether understanding random start, adversarial order variants could help solve the mystery.

\bibliographystyle{plain}
\bibliography{sosa}

\begin{thebibliography}{10}

\bibitem{assadi}
Sepehr Assadi, Yu~Chen, and Sanjeev Khanna.
\newblock Sublinear algorithms for ({$\Delta$} {+} 1) vertex coloring.
\newblock In {\em Proc., ACM-SIAM Symposium on Discrete Algorithms (SODA)},
  pages 767--786, 2019.

\bibitem{barcelo}
J.~Barcelo, B.~Bellalta, C.~Cano, and M.~Oliver.
\newblock Learning-{BEB:} {Avoiding} collisions in {WLANs}.
\newblock {\em Carrier Sense Multiple Access with Enhanced Collision
  Avoidance}, page~23, 2009.

\bibitem{iq}
Apurv Bhartia, Deeparnab Chakrabarty, Krishna Chintalapudi, Lili Qiu, Bozidar
  Radunovic, and Ramachandran Ramjee.
\newblock {IQ}-hopping: Distributed oblivious channel selection for wireless
  networks.
\newblock In {\em Proc., ACM International Symposium on Mobile Ad Hoc
  Networking and Computing (MOBIHOC)}, pages 81--90, 2016.

\bibitem{deepcDistrib}
Sayan Bhattacharya, Deeparnab Chakrabarty, Monika Henzinger, and Danupon
  Nanongkai.
\newblock Dynamic algorithms for graph coloring.
\newblock In {\em Proc., ACM-SIAM Symposium on Discrete Algorithms (SODA)},
  pages 1--20, 2018.

\bibitem{Bhatt19}
Sayan Bhattacharya, Fabrizio Grandoni, Janardhan Kulkarni, Quanquan~C. Liu, and
  Shay Solomon.
\newblock Fully dynamic $({\Delta}+1)$-coloring in constant update time.
\newblock {\em https://arxiv.org/abs/1910.02063}, 2019.

\bibitem{pettie}
Yi-Jun Chang, Wenzheng Li, and Seth Pettie.
\newblock An optimal distributed {($\Delta$+ 1)}-coloring algorithm?
\newblock In {\em Proc., ACM Symposium on Theory of Computing (STOC)}, pages
  445--456, 2018.

\bibitem{checco}
Alessandro Checco and Doug~J. Leith.
\newblock Fast, responsive decentralized graph coloring.
\newblock {\em IEEE/ACM Trans. Netw.}, 25(6):3628--3640, December 2017.

\bibitem{paul}
Paul~B. de~Supinski.
\newblock Convergence times of decentralized graph coloring algorithms.
\newblock Technical Report TR2019-864, Dartmouth College, Computer Science,
  Hanover, NH, May 2019.

\bibitem{duffy2}
Ken~R. Duffy, Charles Bordenave, and Douglas~J. Leith.
\newblock Decentralized constraint satisfaction.
\newblock {\em IEEE/ACM Trans. Netw.}, 21(4):1298--1308, August 2013.

\bibitem{duffy}
Ken~R. Duffy, Neil O'Connell, and Art\"em Sapozhnikov.
\newblock Complexity analysis of a decentralised graph colouring algorithm.
\newblock {\em Inform.\ Process.\ Lett.}, 107(2):60--63, July 2008.

\bibitem{fitz}
Stephen Fitzpatrick and Lambert Meertens.
\newblock Soft, real-time, distributed graph coloring using decentralized,
  synchronous, stochastic, iterative-repair, anytime algorithms.
\newblock {\em Kestrel Institute Technical Report}, May 2001.

\bibitem{galan}
Severino~F. Gal{\'a}n.
\newblock Simple decentralized graph coloring.
\newblock {\em Computational Optimization and Applications}, 66(1):163--185,
  Jan 2017.

\bibitem{HSS16}
David~G. Harris, Johannes Schneider, and Hsin-Hao Su.
\newblock Distributed {$({\Delta} + 1)$}-coloring in sublogarithmic rounds.
\newblock In {\em Proc., ACM Symposium on Theory of Computing (STOC)}, pages
  465--478, 2016.

\bibitem{HP19}
Monika Henzinger and Pan Peng.
\newblock Constant-time dynamic {$({\Delta}+ 1)$}-coloring and weight
  approximation for minimum spanning forest: Dynamic algorithms meet property
  testing.
\newblock {\em arXiv preprint arXiv:1907.04745}, 2019.

\bibitem{leith}
Douglas Leith and Peter Clifford.
\newblock Convergence of distributed learning algorithms for optimal wireless
  channel allocation.
\newblock In {\em Proc., IEEE Conference on Decision and Control (CDC)}, pages
  2980 -- 2985, 01 2007.

\bibitem{linial}
Nathan Linial.
\newblock Locality in distributed graph algorithms.
\newblock {\em SIAM Journal on Computing (SICOMP)}, 21(1):193--201, 1992.

\bibitem{luby}
Michael Luby.
\newblock A simple parallel algorithm for the maximal independent set problem.
\newblock {\em SIAM Journal on Computing (SICOMP)}, 15(4):1036--1053, 1986.

\bibitem{probAndComputing}
Michael Mitzenmacher and Eli Upfal.
\newblock {\em Probability and Computing: Randomized Algorithms and
  Probabilistic Analysis}.
\newblock Cambridge University Press, New York, NY, USA, 2005.

\bibitem{mt}
Robin~A. Moser and G{\'a}bor Tardos.
\newblock A constructive proof of the general {L}ov{\'a}sz local lemma.
\newblock {\em Journal of the ACM}, 57(2):11, 2010.

\bibitem{roughgarden}
Arik Motskin, Tim Roughgarden, Primoz Skraba, and Leonidas Guibas.
\newblock Lightweight coloring and desynchronization for networks.
\newblock In {\em Proc., IEEE INFOCOM}, pages 2383--2391. IEEE, 2009.

\bibitem{johansson}
\"{O}jvind Johansson.
\newblock Simple distributed $({\Delta} {+} 1)$-coloring of graphs.
\newblock {\em Inform.\ Process.\ Lett.}, 70:229--232, 1999.

\bibitem{wald}
Abraham Wald.
\newblock Some generalizations of the theory of cumulative sums of random
  variables.
\newblock {\em Ann. Math. Statist.}, 16(3):287--293, 09 1945.

\end{thebibliography}

\end{document}